\documentclass[11pt,reqno,oneside]{article}
\usepackage{cmap}
\usepackage[utf8]{inputenc}
\usepackage[T2A]{fontenc}
\usepackage[english,russian]{babel}
\usepackage[margin=2.75cm]{geometry}
\usepackage{indentfirst}
\usepackage{microtype} 

\usepackage{cite} 
\usepackage[linesnumbered,vlined,ruled]{algorithm2e} 

\usepackage{tikz}
\usetikzlibrary{babel, arrows.spaced, arrows.meta, bending}

\usepackage{amsmath, amssymb, amsthm}
\usepackage{bm} 

\newcommand{\Z}{\mathbb{Z}}
\DeclareMathOperator{\sgn}{sgn}
\DeclareMathOperator{\conv}{conv}
 
\DeclareMathOperator{\len}{len} 
\DeclareMathOperator{\I}{I} 
\DeclareMathOperator{\J}{J} 

\newcommand{\ot}{\coloneqq}

\renewcommand{\le}{\leqslant}

\renewcommand{\ge}{\geqslant}

\renewcommand{\emptyset}{\varnothing}
\renewcommand{\epsilon}{\varepsilon}

\usepackage{mathtools} 
\providecommand\given{} 
\newcommand\SetSymbol[1][]{%
	\nonscript\:#1\vert
	\allowbreak \nonscript\:	\mathopen{}}
\DeclarePairedDelimiterX\Set[1]\{\}{%
	\renewcommand\given{\SetSymbol[\delimsize]}	#1}

\newtheorem{lemma}{Лемма}
\theoremstyle{definition}
\newtheorem{definition}{Определение}

\usepackage[unicode, colorlinks, linkcolor=blue]{hyperref}

\begin{document}

\hypersetup{pdfauthor={A.N. Maksimenko}, pdftitle={Branch and bound algorithm for the traveling salesman problem is not a direct-type algorithm}}
\title{Алгоритм ветвей и границ для задачи коммивояжера\\ не является алгоритмом прямого типа}
\author{А.Н. Максименко\thanks{Работа выполнена в рамках гос. задания на НИР ЯрГУ, шифр	1.5768.2017/П220.}}

\maketitle


\begin{abstract}
В настоящей работе рассматривается понятие линейного разделяющего алгоритма прямого типа, введенное В.А.~Бондаренко в~1983~г.
До недавнего времени считалось, что класс алгоритмов прямого типа является широким и включает в себя многие классические комбинаторные алгоритмы, в~том числе, алгоритм ветвей и границ для задачи коммивояжера, предложенный J.D.C.~Little, K.G.~Murty, D.W.~Sweeney, C.~Karel в 1963~г.
Мы покажем, что этот алгоритм не является алгоритмом прямого типа.
\end{abstract}

\section{Введение}

В 2015--2018 гг. было опубликовано несколько работ~\cite{Bondarenko:2015,Nikolaev:2016,Nikolaev:2017,Bondarenko:2017,Bondarenko:2018}, основными результатами которых являются оценки кликовых чисел графов многогранников, ассоциированных с~различными задачами комбинаторной оптимизации.
Основной мотивацией для таких оценок является следующий тезис:
\foreignlanguage{english}{``It is known that this value characterizes the time complexity in a broad class of algorithms based on linear comparisons''}%
\footnote{<<Известно, что эта величина характеризует сложность по времени в широком классе алгоритмов, основанных на линейных сравнениях>>}~\cite{Bondarenko:2018}.
А именно, речь идет о классе алгоритмов прямого типа, впервые введенном в~\cite{Bondarenko:1983}.
В качестве подтверждения этого тезиса в~\cite{Nikolaev:2016,Nikolaev:2017} говорится о том, что этот класс
включает алгоритмы сортировки, жадный алгоритм, динамическое программирование и~метод ветвей и~границ\footnote{Но ссылки на источник с соответствующими доказательствами не приводятся.}.
Доказательства того, что эти алгоритмы (а также алгоритм Эдмондса для задачи о паросочетаниях)
являются алгоритмами прямого типа, впервые были опубликованы в диссертации~\cite{Bond:1993} (см. также монографию~\cite{BondBook:2008}).
В 2014 г. в~\cite{Maksimenko:2014} было показано, что алгоритм Куна"--~Манкреса для задачи о~назначениях (а~вместе с ним и алгоритм Эдмондса) не принадлежит к этому классу. Там же был описан часто используемый на практике способ модификации алгоритмов, выводящий их из класса алгоритмов прямого типа.
Ниже мы докажем, что классический алгоритм ветвей и~границ для задачи коммивояжера~\cite{Little:1963,Reingold:1980} тоже не принадлежит к этому классу.
Тем~самым будет показано, что
теорема~2.6.3 из диссертации~\cite{Bond:1993} (теорема~3.6.6 из многографии~\cite{BondBook:2008})
не~может быть доказана в оригинальной постановке.
Это позволяет сделать вывод о том, что класс алгоритмов прямого типа не является столь широким, как предполагалось ранее.
 
Текст статьи организован следующим образом. 
В разделе~\ref{sec:Alg} приводится псевдокод классического алгоритма ветвей и границ для задачи коммивояжера.
В разделе~\ref{sec:Direct} вводятся основные понятия концепции алгоритмов прямого типа и два ключевых определения: алгоритма прямого типа и алгоритма <<прямого типа>>.
В разделе~\ref{sec:notDirect1} показано, что классический алгоритм ветвей и границ для задачи коммивояжера не является алгоритмом прямого типа, а в разделе~\ref{sec:notDirect2} "--- что он не является алгоритмом <<прямого типа>>.

%
%

\section{Алгоритм ветвей и границ для задачи коммивояжера}
\label{sec:Alg}

Рассмотрим полный орграф $G = (V, A)$ с множеством вершин $V = [n] = \{1,2,\dots,n\}$ и~дуг $A = \Set{(i,j) \given i,j \in V, \ i \ne j}$. Каждой дуге $(i,j) \in A$ поставлено в соответствие число $c_{ij} \in \Z$, называемое \emph{длиной дуги}. 
\emph{Длиной подмножества} $H \subseteq A$ будем называть суммарную длину входящих в него дуг: $\len(H) = \sum_{(i,j) \in H} c_{ij}$.
Задача коммивояжера состоит в том, чтобы найти $H^* \subseteq A$, являющееся гамильтоновым контуром в $G$ и имеющее минимальную длину $\len(H^*)$.

Для удобства дальнейшего обсуждения поместим числа $c_{ij}$ в матрицу $C = (c_{ij})$. 
Диагональным элементам $c_{ii}$ припишем максимально возможные длины, $c_{ii}  \ot \infty$,
чтобы исключить их влияние на работу алгоритма, и будем предполагать, что $\infty - b = \infty$ для любого числа $b \in \Z$.
Через $\I(M)$ будем обозначать множество индексов строк матрицы $M$,
а~через $\J(M)$ обозначим множество индексов столбцов матрицы $M$.
В начале работы алгоритма $\I(C) = \J(C) = V$.
Через $M(S,T)$ обозначим подматрицу матрицы $M$, лежащую на пересечении строк $S \subseteq \I(M)$ и столбцов $T \subseteq \J(M)$.

Сам алгоритм подробно описан в \cite[раздел 4.1.6]{Reingold:1980} и \cite{Little:1963}.
Мы приводим лишь его псевдокод "--- алгоритм~\ref{alg}.
Отдельно, в алгоритме~\ref{alg:reduction} описан процесс редуцирования строк и столбцов матрицы, а в алгоритме~\ref{alg:arc} "--- способ выбора такого нулевого элемента матрицы, при замене которого на бесконечность сумма редукций матрицы максимальна.

\SetAlgorithmName{Алгоритм}{Список алгоритмов}{} 
\SetAlgoCaptionSeparator{.} 
\DontPrintSemicolon 
\SetKwProg{Proc}{Procedure}{}{end} 
\SetKwProg{Fn}{Function}{}{end} 
\SetKwInOut{Global}{Глобальные}
\SetKwInOut{Input}{Вход}
\SetKwInOut{Output}{Выход}
\SetKwRepeat{DoWhile}{do}{while} 
\SetKwBlock{Loop}{loop}{endloop} 
\SetKwFunction{Reduction}{Reduction}
\SetKwFunction{BranchBound}{BranchBound}
\SetKwFunction{HamiltonianCycle}{HamiltonCycle}
\SetKwFunction{ForbiddenArc}{ForbiddenArc}
\SetKwFunction{ChooseArc}{ChooseArc}
\SetKwArray{Arcs}{Arcs} 
\SetKwArray{Mat}{M} 
\SetKwArray{NewMat}{Mnew} 
\SetKwData{Sum}{sum} 
\SetKwArray{Tour}{Hopt} 
\SetKwData{Value}{lopt} 

\begin{algorithm}
	\caption{Метод ветвей и границ для задачи коммивояжера} 
	\label{alg}
	\Global{ гамильтонов контур \Tour с минимальной длиной; его длина \Value. До начала работы алгоритма $\Value \ot \infty$.}
	\Input{ матрица длин \Mat; множество дуг \Arcs, обязательных для включения в контур; текущая сумма всех редукций \Sum. В~самом начале работы алгоритма $\Mat \ot C$, $\Arcs \ot \emptyset$, $\Sum \ot 0$.}
	\BlankLine
	\Proc{\BranchBound{\Mat, \Arcs, \Sum}}{
		\tcc{Редуцируем матрицу \Mat}
		\Reduction{\Mat, \Sum}\; 
		\If{$\Sum \ge \Value$}{\label{alg:return}
			завершить текущий экземпляр процедуры
		}
		\tcc{Выбираем оптимальный нулевой элемент матрицы \Mat}
		$(i,j) \ot {}$\ChooseArc{\Mat} \label{alg:zero} \;
		\tcc{Разбираем случаи, когда контур содержит дугу $(i,j)$}
		\eIf{$|I| = 3$}{
			\tcc{Находим единственный гамильтонов контур}
			$H \ot {}$\HamiltonianCycle{$\Arcs \cup \{(i,j)\}$}\;
			\If{$\len(H) < \Value$}{
				$\Tour \ot H$\;
				$\Value \ot \len(H)$\;
			}
		}{
			\tcc{Вычеркиваем $i$-ю строку и $j$-й столбец}
			$\NewMat \ot \Mat(\I(\Mat) \setminus \{i\}, \J(\Mat) \setminus \{j\})$\;
			\tcc{Находим запрещенную дугу}
			$(l,k) \ot \ForbiddenArc{\Arcs,(i,j)}$\;
			$\NewMat{l,k} \ot \infty$\;
			\BranchBound{\NewMat, $\Arcs \cup \{(i,j)\}$, \Sum} \label{alg:bb2} \;
		}
		\tcc{Разбираем случаи, когда контур не содержит дугу $(i,j)$}
		$\Mat{i,j} \ot \infty$\;
		\BranchBound{\Mat, \Arcs, \Sum} \label{alg:bb3} \;
	}
	\BlankLine
	\Fn{\HamiltonianCycle{\Arcs}}{
		Найти гамильтонов контур, содержащий все дуги из \Arcs.\;
	}	
	\BlankLine
	\Fn{\ForbiddenArc{\Arcs,$(i,j)$}}{
		Найти пару вершин $l$ и $k$, являющихся концом и началом наибольшего (по включению) пути в \Arcs, содержащего $(i,j)$.
	}	
\end{algorithm}

\begin{algorithm}
	\caption{Редуцирование строк и столбцов матрицы} 
	\label{alg:reduction}
	\Input{ матрица \Mat; текущая сумма всех редукций \Sum.}
	\Output{ редуцированная матрица \Mat; измененная \Sum.}
	\BlankLine
	\Proc{\Reduction{\Mat, \Sum}}{
		\tcc{Редуцируем строки матрицы \Mat}
		\For{$i \in \I(\Mat)$}{
			$m \ot \infty$\;
			\tcc{Находим $m = m(i) = \min_{j \in \J(\Mat)} \Mat{i,j}$}
			\For{$j \in \J(\Mat)$}{
				\lIf{$m > \Mat{i,j}$}{$m \ot \Mat{i,j}$} \label{alg:min}
			}	
			$\Sum \ot \Sum + m$\;
			\lFor{$j \in \J(\Mat)$}{$\Mat{i,j} \ot \Mat{i,j} - m$}
		}	
		\tcc{Редуцируем столбцы матрицы \Mat}
		\For{$j \in \J(\Mat)$}{
			$m \ot \infty$\;
			\For{$i \in \I(\Mat)$}{
				\lIf{$m > \Mat{i,j}$}{$m \ot \Mat{i,j}$}
			}	
			$\Sum \ot \Sum + m$\;
			\lFor{$i \in \I(\Mat)$}{$\Mat{i,j} \ot \Mat{i,j} - m$}	
		}	
	}
\end{algorithm}

\begin{algorithm}
	\caption{Выбор дуги} 
	\label{alg:arc}
	\Input{ матрица \Mat.}
	\Output{ дуга $(i^*,j^*)$, при запрещении которой нижняя оценка длины гамильтонова контура максимальна.}
	\BlankLine
	\Fn{\ChooseArc{\Mat}}{
		$w \ot -1$\;
		\For{$i \in \I(\Mat)$}{
			\For{$j \in \J(\Mat)$}{
				\If{$\Mat{i,j} = 0$}{
					$m \ot \infty$\;
					\tcc{Находим $m = \min_t \Mat{i,t}$}
					\For{$t \in \J(\Mat) \setminus \{j\}$}{
						\lIf{$m > \Mat{i,t}$}{$m \ot \Mat{i,t}$} \label{alg:imin}
					}	
					$k \ot \infty$\;
					\tcc{Находим $k = \min_t \Mat{t,j}$}
					\For{$t \in \I(\Mat) \setminus \{i\}$}{
						\lIf{$k > \Mat{t,j}$}{$k \ot \Mat{t,j}$} \label{alg:jmin}
					}	
					\tcc{Сравниваем $m+k$ с текущим рекордом $w$}
					\If{$m + k > w$\label{alg:mkw}}{
						$w \ot m + k$\;
						$(i^*,j^*) \ot (i,j)$\;
					}	
				}
			}	
		}
	}
\end{algorithm}

%
%

\section{Алгоритмы прямого типа}
\label{sec:Direct}

При изложении основ теории алгоритмов прямого типа мы будем придерживаться~\cite{Bond:1993} (см. также~\cite{BondBook:2008}).

С целью унификации изложения матрица длин дуг $C$ далее будет называться \emph{вектором}\footnote{Элементы матрицы всегда можно выписать в строку или столбец.} \emph{входных данных} или просто \emph{входом}. Решение задачи коммивояжера, т.е. гамильтонов контур $H \subseteq A$, будет представляться в виде 0/1"=вектора $\bm{x} = (x_{ij})$, имеющего ту же размерность, что и $C$. Координаты этого вектора $x_{ij} = 1$, при $(i,j) \in H$, и $x_{ij} = 0$ иначе. Через $X$ обозначаем множество всех 0/1"=векторов $\bm{x}$, соответствующих гамильтоновым контурам в рассматриваемом орграфе $G$. Таким образом, при фиксированном входе $C$ задача коммивояжера состоит в поиске решения $\bm{x^*} \in X$ такого, что $\langle \bm{x^*}, C\rangle \le \langle \bm{x}, C\rangle$ $\forall \bm{x} \in X$.
Далее будем называть такое решение $\bm{x^*}$ \emph{оптимальным относительно входа $C$}.
Следуя~\cite[определение 1.1.2]{Bond:1993}, совокупность всех таких оптимизационных задач, образованную фиксированным множеством допустимых решений $X$ (в случае задачи коммивояжера, $X$ однозначно определяется числом вершин орграфа $G$) и всевозможными входными векторами $C$, будем называть \emph{задачей $X$}.
Два допустимых решения $\bm{x}, \bm{y} \in X$ задачи $X$ называются \emph{смежными}, если найдется вектор $C$ такой, что они, и только они, являются оптимальными относительно~$C$.
Подмножество $Y \subseteq X$ называется \emph{кликой}, если любая пара $\bm{x}, \bm{y} \in Y$ смежна.

Выпуклая оболочка $\conv(X)$ называется \emph{многогранником задачи $X$}. Так как $X$ в задаче коммивояжера является подмножеством вершин единичного куба, то $X$ совпадает с множеством вершин многогранника $\conv(X)$.
В этой терминологии два решения $\bm{x}, \bm{y} \in X$ смежны тогда и только тогда, когда смежны соответствующие вершины многогранника $\conv(X)$~\cite{Bond:1993}.
Известно~\cite{Padberg:1974}, что все вершины многогранника коммивояжера попарно смежны при $n < 6$, где $n$ "--- число вершин орграфа $G$, в котором требуется найти оптимальный гамильтонов контур.

Алгоритмы прямого типа относятся к классу линейных разделяющих алгоритмов, которые удобно представлять в виде линейных разделяющих деревьев.

\begin{definition}[{{\cite[определение 1.3.1]{Bond:1993}}}] 
	\emph{Линейным  разделяющим деревом} задачи $X \subset \Z^m$ называется ориентированное дерево, обладающее следующими свойствами:
	\begin{itemize}
		\item[а)] 
		в каждый узел, за исключением одного, называемого  корнем,
		входит ровно одна дуга; дуг, входящих в корень, нет;
		\item[б)] 
		для каждого узла либо имеется две выходящих из него  дуги,
		либо таких дуг  нет  вообще;   в  первом  случае  узел  называется
		внутренним, во втором "--- внешним, или листом;
		\item[в)] 
		каждому внутреннему узлу соответствует некоторый вектор $B \in \Z^m$;
		\item[г)] 
		каждому листу соответствует некоторый элемент из $X$ (нескольким листьям может соответствовать один и тот же элемент множества $X$);
		\item[д)] 
		каждой дуге $d$ соответствует число  $\sgn d$, равное $1$ либо $-1$;
		две дуги, выходящие из одного узла, имеют различные значения;
		\item[е)] 
		для каждой цепи $W = B_1 d_1 B_2 d_2 \ldots B_k d_k \bm{x}$, соединяющей корень и лист (в обозначении  цепи  перечислены  соответствующие  ее  узлам векторы $B_i$; дуга $d_i$ выходит из узла $B_i$, $i\in[k]$), и для любого входа $C$ из неравенств $\langle B_i, C \rangle \sgn d_i \ge 0$, $i\in[k]$, следует, что решение $\bm{x}$ является оптимальным относительно $C$.
	\end{itemize}
\end{definition}

Таким образом, в рамках теории линейных разделяющих алгоритмов внимание уделяется только тем операциям, где выполняется проверка условий вида $\langle B, C \rangle \ge 0$, где $C$ "--- вектор входных данных. 
Так, например, в строке~\ref{alg:min} алгоритма~\ref{alg:reduction} на самом первом шаге цикла проверяется неравенство $\infty > C_{11}$; на втором шаге проверяется условие $C_{11} > C_{12}$, и~т.\,д.
А в функциях \HamiltonianCycle и \ForbiddenArc, с точки зрения линейных разделяющих алгоритмов, не происходит ничего интересного, так как не выполняются никакие сравнения с элементами вектора входных данных.

Процесс работы линейного разделяющего алгоритма для фиксированного вектора входных данных $C$ представляет собой некоторую цепь $B_1 d_1 B_2 d_2 \ldots B_m d_m \bm{x}$, соединяющую корень $B_1$ и некоторый лист~$\bm{x}$ соответствующего линейного разделяющего дерева.
Листом в нашем случае является гамильтонов контур (точнее, его характеристический вектор), являющийся оптимальным относительно~$C$.

Пусть $B$ "--- некоторый внутренний  узел в линейном разделяющем дереве рассматриваемого алгоритма,
а $X$ "--- множество всех допустимых решений (множество меток всех листьев).
Обозначим через $X_B$, $X_B \subseteq X$,  множество меток всех листьев этого дерева,
которым предшествует  узел $B$, а через $X_B^+$ и $X_B^-$  обозначим
подмножества множества $X_B$, соответствующие  двум выходящим из $B$ дугам.
Очевидно, $X_B = X_B^+ \cup X_B^-$.
Обозначим через $R_B^- = X_B^+ \setminus X_B^-$ множество меток,
отбрасываемых при переходе по <<отрицательной>> дуге.
По аналогии определим множество меток $R_B^+ = X_B^- \setminus X_B^+$,
отбрасываемых при переходе по <<положительной>> дуге.

\begin{definition}[{{\cite[определение 1.4.2]{Bond:1993}}}] 
	\label{def:direct-type}
	Линейное разделяющее дерево называется деревом \emph{прямого типа}, если для любого внутреннего узла $B$ и~для любой клики $Y \subseteq X$ выполняется неравенство
	\begin{equation}
	\min \{ |R_B^+ \cap Y|, |R_B^- \cap Y| \} \le 1.
	\label{eq:direct-type}
	\end{equation}
\end{definition}

Непосредственно из определения следует, что высота дерева прямого типа (то есть число сравнений, используемых алгоритмом в худшем случае) для задачи $X$ не может быть меньше, чем $\omega(X) - 1$, где $\omega(X)$ "--- кликовое число множества $X$~\cite[теорема 1.4.3]{Bond:1993}.

Если же мы хотим доказать, что некий алгоритм не является алгоритмом прямого типа, достаточно указать клику $Y$, состоящую из четырех решений, и узел $B$ такие, что $|R_B^+ \cap Y| = |R_B^- \cap Y| = 2$.

Для каждого $\bm{x} \in X$ определим \emph{конус исходных данных}
\begin{equation*}
K(\bm{x}) = \Set*{C \given  \langle \bm{x}, C\rangle \le \langle \bm{y}, C\rangle, \ \forall \bm{y} \in X}.
\end{equation*}
Т.\,е. $K(\bm{x})$ состоит из всех векторов $C$ таких, что $\bm{x}$ оптимален относительно $C$.

\begin{definition}[{{\cite[определение 1.4.4]{Bond:1993}}}] 
	\label{def:direct-type2}
	Линейное разделяющее дерево называется деревом \emph{<<прямого типа>>}, если 
	каждая цепь $B_1 d_1 B_2 d_2 \ldots B_k d_k \bm{x}$,  
	соединяющая  корень  и  лист, удовлетворяет условиям:
	\begin{itemize}
		\item[(*)] 
		для любого $\bm{y}\in X$, смежного с $\bm{x}$, найдется такой номер $i\in [k]$, 
		что условия $\langle B_i, C \rangle \sgn d_i > 0$ и $C\in K(\bm{y})$ несовместны;
		\item[(**)] 
		для любого $i\in[k]$ из несовместности условий
		\[
		\langle B_i, C \rangle \sgn d_i > 0 \qquad \mbox{и} \qquad C\in K(\bm{y})
		\]
		для $\bm{y}$, смежного с $\bm{x}$, и из телесности конуса
		\[
		K(\bm{x}) \cap \Set*{ C \given \langle B_i, C \rangle \sgn d_i \le 0}
		\]
		следует, что ветвь, начинающаяся в узле $B_i$ с дугой $-d_i$,  имеет хотя бы один лист, помеченный $\bm{x}$.
	\end{itemize}
\end{definition}

Деревья <<прямого типа>> с деревьями прямого типа объединяет тот факт, что их высота тоже ограничена снизу величиной $\omega(X) - 1$~\cite[теорема~1.4.5]{Bond:1993}.

Чтобы доказать, что алгоритм~\ref{alg} не является алгоритмом <<прямого типа>>, мы ограничимся проверкой условия (*) из этого определения. А~именно, мы укажем вполне конкретный входной вектор $C^*$, который однозначно определит некоторую цепь $B_1 d_1 B_2 d_2 \ldots B_k d_k \bm{x}$.
Далее будет выбран $\bm{y}\in X$, смежный с $\bm{x}$, для которого условия $\langle B_i, C \rangle \sgn d_i > 0$ и $C\in K(\bm{y})$ совместны при любом $i\in [k]$. Обратим особое внимание на то, что нам нужно будет проверить совместность условий $\langle B_i, C \rangle \sgn d_i > 0$ и $C \in K(\bm{y})$ отдельно для каждого $i\in [k]$, вне зависимости от результатов других сравнений.
То есть для каждого $i \in [k]$ достаточно указать $C_i$ такой, что $\langle B_i, C_i \rangle \sgn d_i > 0$ и $C_i \in K(\bm{y})$.

%
%

\section{Алгоритм~\ref{alg} не является прямым}
\label{sec:notDirect1}

Рассмотрим задачу коммивояжера в полном орграфе на 5 вершинах.
Множество допустимых решений $X$ такой задачи состоит из двадцати четырех 0/1-векторов, соответствующих гамильтоновым контурам в этом орграфе. Все 24 решения попарно смежны~\cite{Padberg:1974}.

Предположим, что элементы матрицы длин дуг $C \in \Z^{5\times 5}$ удовлетворяют следующим условиям:
\begin{equation}
\label{eq:cond}
\begin{aligned}
c_{12} &\le c_{13},  &  c_{12} &\le c_{14},  &  c_{12} &\le c_{15}, \\ 
c_{21} &\le c_{23},  &  c_{21} &\le c_{24},  &  c_{21} &\le c_{25}, \\ 
c_{31} &> c_{32},    &  c_{32} &> c_{34},    &  c_{34} &> c_{35}. \\ 
\end{aligned}
\end{equation}
В самом начале работы рассматриваемого алгоритма выполняется процедура редуцирования этой матрицы (алгоритм~\ref{alg:reduction}).
Мы ограничимся рассмотрением этапа редуцирования строк.
В результате последовательных сравнений в первой строке выбирается наименьший элемент
(в данном случае $c_{12}$) и вычитается из всех её элементов.
Далее выбирается минимальный элемент во второй строке, им оказывается $c_{21}$,
и минимальный элемент в третьей строке "--- $c_{35}$.
После этого алгоритм переходит к проверке неравенства
\begin{equation}
\label{eq:B}
c_{41} > c_{42}
\end{equation}
(сравнение $\infty > c_{41}$ присутствует в алгоритме исключительно для краткости описания и не несет никакой информации).
Соответствующий узел линейного разделяющего дерева алгоритма обозначим $B$.
Ясно, что алгоритм попадает в этот узел дерева, если, и только если для входного вектора $C$ выполняются условия~\eqref{eq:cond}.

Рассмотрим характеристические вектора четырех гамильтоновых контуров:
\[
\begin{aligned}
\bm{x} &= 
\begin{pmatrix}
   & 0 & 0 & 1 & 0 \\
 0 &   & 0 & 0 & 1 \\
 1 & 0 &   & 0 & 0 \\
 0 & 1 & 0 &   & 0 \\
 0 & 0 & 1 & 0 &  \\
\end{pmatrix},
&
\bm{y} &= 
\begin{pmatrix}
   & 0 & 0 & 0 & 1 \\
 0 &   & 1 & 0 & 0 \\
 1 & 0 &   & 0 & 0 \\
 0 & 1 & 0 &   & 0 \\
 0 & 0 & 0 & 1 &  \\
\end{pmatrix},
\\
\bm{z} &= 
\begin{pmatrix}
   & 0 & 1 & 0 & 0 \\
 0 &   & 0 & 0 & 1 \\
 0 & 1 &   & 0 & 0 \\
 1 & 0 & 0 &   & 0 \\
 0 & 0 & 0 & 1 &  \\
\end{pmatrix},
&
\bm{w} &= 
\begin{pmatrix}
   & 0 & 0 & 0 & 1 \\
 0 &   & 0 & 1 & 0 \\
 0 & 1 &   & 0 & 0 \\
 1 & 0 & 0 &   & 0 \\
 0 & 0 & 1 & 0 &  \\
\end{pmatrix}.
\end{aligned}
\]
Нетрудно проверить, что входные векторы
\[
\begin{aligned}
C_x &= 
\begin{pmatrix}
  & 0 & 6 & 1 & 6 \\
0 &   & 6 & 6 & 1 \\
3 & 2 &   & 1 & 0 \\
6 & 0 & 6 &   & 6 \\
6 & 6 & 0 & 6 &  \\
\end{pmatrix},
&
C_y &= 
\begin{pmatrix}
& 0 & 6 & 6 & 1 \\
0 &   & 1 & 6 & 6 \\
3 & 2 &   & 1 & 0 \\
6 & 0 & 6 &   & 6 \\
6 & 6 & 6 & 0 &  \\
\end{pmatrix},
\\
C_z &= 
\begin{pmatrix}
& 0 & 1 & 6 & 6 \\
0 &   & 6 & 6 & 1 \\
6 & 3 &   & 1 & 0 \\
0 & 6 & 6 &   & 6 \\
6 & 6 & 6 & 0 &  \\
\end{pmatrix},
&
C_w &= 
\begin{pmatrix}
& 0 & 6 & 6 & 1 \\
0 &   & 6 & 1 & 6 \\
6 & 3 &   & 1 & 0 \\
0 & 6 & 6 &   & 6 \\
6 & 6 & 0 & 6 &  \\
\end{pmatrix}
\end{aligned}
\]
удовлетворяют условиям~\eqref{eq:cond}, а для каждого $\bm{t} \in \{\bm{x}, \bm{y}, \bm{z}, \bm{w}\}$ и для любого $\bm{s} \in X \setminus \{\bm{t}\}$ выполняется неравенство $\langle \bm{t}, C_{t}\rangle = 5 < \langle \bm{s}, C_{t}\rangle$.
Следовательно, все четыре вектора входят в множество меток $X_B$ всех листьев дерева алгоритма, которым предшествует узел $B$.

Покажем, что $\bm{z}$ и $\bm{w}$ входят в множество меток $R_B^+$, отбрасываемых при выполнении неравенства~\eqref{eq:B}, а $\bm{x}$ и $\bm{y}$ входят в множество меток $R_B^-$, отбрасываемых при невыполнении неравенства~\eqref{eq:B}.

Предположим, что для входной матрицы $C$ выполнены условия~\eqref{eq:cond} и неравенство~\eqref{eq:B}. 
Тогда $\langle \bm{z}, C\rangle > \langle \bm{z'}, C\rangle$ для
\[
\bm{z'} = 
\begin{pmatrix}
  & 0 & 1 & 0 & 0 \\
1 &   & 0 & 0 & 0 \\
0 & 0 &   & 0 & 1 \\
0 & 1 & 0 &   & 0 \\
0 & 0 & 0 & 1 &  \\
\end{pmatrix}.
\]
Аналогично, $\langle \bm{w}, C\rangle > \langle \bm{w'}, C\rangle$ для
\[
\bm{w'} = 
\begin{pmatrix}
  & 0 & 0 & 0 & 1 \\
1 &   & 0 & 0 & 0 \\
0 & 0 &   & 1 & 0 \\
0 & 1 & 0 &   & 0 \\
0 & 0 & 1 & 0 &  \\
\end{pmatrix}.
\]
Таким образом, $\bm{z}, \bm{w} \in R_B^+$.

Предположим, что для $C$ выполнены условия~\eqref{eq:cond}, но не выполнено неравенство~\eqref{eq:B}. 
Тогда $\langle \bm{x}, C\rangle > \langle \bm{x'}, C\rangle$ для
\[
\bm{x'} = 
\begin{pmatrix}
  & 1 & 0 & 0 & 0 \\
0 &   & 0 & 0 & 1 \\
0 & 0 &   & 1 & 0 \\
1 & 0 & 0 &   & 0 \\
0 & 0 & 1 & 0 &  \\
\end{pmatrix},
\]
и $\langle \bm{y}, C\rangle > \langle \bm{y'}, C\rangle$ для
\[
\bm{y'} = 
\begin{pmatrix}
  & 1 & 0 & 0 & 0 \\
0 &   & 1 & 0 & 0 \\
0 & 0 &   & 0 & 1 \\
1 & 0 & 0 &   & 0 \\
0 & 0 & 0 & 1 &  \\
\end{pmatrix}.
\]
следовательно, $\bm{z}, \bm{w} \in R_B^+$.

Таким образом, условие~\eqref{eq:direct-type} для данного узла $B$ не выполнено, и алгоритм~\ref{alg} не является алгоритмом прямого типа.


%
%

\section{Алгоритм~\ref{alg} не является <<прямым>>}
\label{sec:notDirect2}

При анализе алгоритма~\ref{alg}, как линейного разделяющего дерева, нам будут встречаться только неравенства следующего вида:
\begin{equation}
\label{eq:BC}
\langle B^+, C \rangle - \langle B^-, C \rangle > 0,
\end{equation}
где $C \in \Z^{n^2}$ "--- вектор входных данных,
\begin{equation}
\label{eq:BC2}
B^+, B^- \in \{0,1\}^{n^2}, \quad
\langle B^+, B^- \rangle = 0 \quad \text{и} \quad
\langle B^+, \bm{1} \rangle = \langle B^-, \bm{1} \rangle > 0,
\end{equation}
$\bm{1}$ "--- вектор из единиц.
Иными словами, условие~\eqref{eq:BC2} означает, что множества единичных координат для $B^+$ и $B^-$ равномощны и не пересекаются.
Для каждого такого неравенства и для некоторого допустимого решения $\bm{y} \in X \subset \{0,1\}^{n^2}$ нам нужно будет проверить, что существует $C \in K(\bm{y})$, для которого это неравенство выполнено.
Такой анализ существенно упрощается, если воспользоваться следующим критерием.

\begin{lemma}
\label{lem:1}
Пусть $\bm{y} \in \{0,1\}^{n^2}$ "--- характеристический вектор некоторого гамильтонова контура в полном орграфе $G = ([n], A)$. Если выполняются условия~\eqref{eq:BC2} и~$\langle B^+, \bm{y} \rangle \le 2$, то неравенство~\eqref{eq:BC} и условие $C \in K(\bm{y})$ совместны.
\end{lemma}
\begin{proof}
Пусть
\[
S = \Set{(i,j) \in [n]^2 \given y_{ij} = 1 \text{ и } B^+_{ij} = 0}.
\]
Из условия $\langle B^+, \bm{y} \rangle \le 2$ следует, что $|S| \ge n-2$.
Положим 
\[
C \ot \bm{4} - B^-
\]
и, после этого, 
$C_{ij} \ot 0$ для $(i,j) \in S$.
Тогда $\langle B^+, C \rangle = \langle B^+, \bm{4} - B^- \rangle = \langle B^+, \bm{4} \rangle$
и~$\langle B^-, C \rangle \le \langle B^-, \bm{4} - B^- \rangle = \langle B^+, \bm{4} \rangle - \langle B^-, B^- \rangle$ (так как $B^+$ и $B^-$ удовлетворяют условиям~\eqref{eq:BC2}).
Следовательно, неравенство~\eqref{eq:BC} для такого $C$ будет выполнено.

Покажем теперь, что $\langle \bm{y}, C \rangle < \langle \bm{x}, C \rangle$ для любого $\bm{x} \in X \setminus \bm{y}$.

Очевидно, $\langle \bm{y}, C \rangle = (n-|S|) 4 \le 8$.

Пусть $\bm{x} \in X$. Заметим, что если $\langle \bm{y}, \bm{x} \rangle \ge n-2$, то $\bm{x} = \bm{y}$, 
так как любой гамильтонов контур в орграфе на $n$ вершинах однозначно определяется по любым своим $n-2$ дугам.
Следовательно, $\langle \bm{x}, C \rangle \ge 3 \cdot 3 = 9$ для любого $\bm{x} \in X \setminus \bm{y}$.
\end{proof}

В частности, условия леммы выполнены, если в $B^+$ не более двух единиц.

Итак, положим $n = 4$ и рассмотрим следующий вектор входных данных (вместо бесконечности будем подставлять пробел):
\begin{equation}
\label{eq:C}
C^* \ot
\begin{pmatrix}
   & 0 & 2 & 1 \\
 2 &   & 0 & 2 \\
 1 & 2 &   & 0 \\
 0 & 1 & 2 &   
\end{pmatrix}
\end{equation}
Ясно, что единственным оптимальным решением будет вектор
\[
\bm{x} \ot
\begin{pmatrix}
  & 1 & 0 & 0 \\
0 &   & 1 & 0 \\
0 & 0 &   & 1 \\
1 & 0 & 0 &   
\end{pmatrix}
\]
и соответствующий ему контур $\{(1,2), (2,3), (3,4), (4,1)\}$.
Нетрудно проверяется, что множество всех допустимых решений $X$ состоит из 6 попарно смежных векторов.
Положим
\[
\bm{y} \ot
\begin{pmatrix}
  & 0 & 0 & 1 \\
0 &   & 1 & 0 \\
1 & 0 &   & 0 \\
0 & 1 & 0 &   
\end{pmatrix}
\]
Обратим внимание, что $\bm{y}$ является вторым (после $\bm{x}$) по оптимальности относительно $C^*$.
Именно это обстоятельство во многом упрощает дальнейшую проверку соответствующих сравнений.

В целом схема работы алгоритма при заданном входе $C^*$ изображена на рис.~\ref{fig:scheme}.
\begin{figure}
\begin{tikzpicture}[yscale=0.6, >=stealth']
\node[rectangle, draw, right] (B1) at (0, 0) 
	{\textbf{1.} \BranchBound{$C^*$, $\emptyset$, 0}};
\draw (B1.north west) -- (0,-10.5) -- +(1,0);
\draw (0,-1) node[right] {\Reduction{$C^*$, \Sum}}
      (0,-2) node[right] {. . . . . .}
      (0,-3) node[right] {\BranchBound{}}
      (0,-9) node[right] {$\Mat{2,3} \ot \infty$}
      (0,-10) node[right] {\BranchBound{}};
\draw[->]  (2.8,-3) -- +(0.8,0);
\begin{scope}[xshift=3.8cm, yshift=-3cm]
\node[rectangle, draw, right] (B2) at (0, 0) 
	{\textbf{2.} \BranchBound{$C'$, $(2,3)$, 0}};
\draw (B2.north west) -- (0,-5.5) -- +(1,0);
\draw (0,-1) node[right] {. . . . . .}
	  (0,-2) node[right] {$H \ot {}$\HamiltonianCycle{}}
	  (0,-3) node[right] {. . . . . .}
	  (0,-4) node[right] {$\Mat{1,2} \ot \infty$}
	  (0,-5) node[right] {\BranchBound{}};
\draw[->]  (2.8,-5) -- +(0.8,0);
\begin{scope}[xshift=3.8cm, yshift=-5cm]
\node[rectangle, draw, right] (B3) at (0, 0) 
	{\textbf{3.} \BranchBound{$C''$, $(2,3)$, 0}};
\draw (B3.north west) -- (0,-3.5) -- +(1,0);
\draw (0,-1) node[right] {\Reduction{$C''$, \Sum}}
	  (0,-2) node[right] {\lIf{$\Sum \ge \Value$}{}}
	  (0.5,-3) node[right] {завершить процедуру};
\end{scope}
\end{scope}
\draw[->, rounded corners = 3mm] (7.4,-11) -- ++(-0.5, 0) -- ++(0, 2) -- +(-4.1,0);
\draw[->, rounded corners = 3mm]  (2.8,-10) -- ++(1.5,0) -- +(0,-0.7);
\begin{scope}[xshift=2.0cm, yshift=-11.5cm]
\node[rectangle, draw, right] (B4) at (0, 0) 
	{\textbf{4.} \BranchBound{$C'''$, $\emptyset$, 0}};
\draw (B4.north west) -- (0,-3.5) -- +(1,0);
\draw (0,-1) node[right] {\Reduction{$C'''$, \Sum}}
(0,-2) node[right] {\lIf{$\Sum \ge \Value$}{}}
(0.5,-3) node[right] {завершить процедуру};
\end{scope}
\end{tikzpicture}
\caption{Общая схема работы алгоритма~\protect\ref{alg} при входе, задаваемом формулой~\protect\eqref{eq:C}}
\label{fig:scheme}
\end{figure}

Рассмотрим, прежде всего, какие неравенства проверяются при первом входе в процедуру \BranchBound с входом $C^*$. При редуцировании первой строки матрицы~$C^*$ (строка~\ref{alg:min} алгоритма~\ref{alg:reduction}) проверяются (и выполняются) неравенства $\infty > C_{12}$, $C_{13} > C_{12}$ и $C_{14} > C_{12}$. 
Далее мы не будем рассматривать неравенства, в которых сумма (либо разность) элементов исходной матрицы сравнивается с бесконечностью, так как они всегда выполняются и совместны с любым допустимым решением.
Заметим, что только что перечисленные неравенства удовлетворяют условиям леммы~\ref{lem:1}, так как $\langle B^+, \bm{1}\rangle = 1$.
А значит, они совместны с условием $C \in K(\bm{y})$.

После редуцирования первой строки в её ячейках $\Mat[1,j]$, $j\in[4]$, содержатся разности $C_{1j} - C_{12}$, а переменная \Sum принимает значение $C_{12}$.

При редуцировании второй строки проверяются неравенства  $C_{21} > C_{23}$ и $C_{24} > C_{23}$.
Согласно лемме~\ref{lem:1}, они совместны с условием $C \in K(\bm{y})$.

После редуцирования второй строки в её ячейках $\Mat[2,j]$, $j\in[4]$, содержатся разности $C_{2j} - C_{23}$, а переменная \Sum принимает значение $C_{12} + C_{23}$.

При редуцировании последних двух строк ситуация полностью аналогична.
После завершения редуцирования строк
\[
\Sum =  C_{12} + C_{23} + C_{34} + C_{41},
\]
\[
\Mat =
\begin{pmatrix}
            &           0 & C_{13} - C_{12} & C_{14} - C_{12} \\
C_{21} - C_{23} &             &           0 & C_{24} - C_{23} \\
C_{31} - C_{34} & C_{32} - C_{34} &             & 0 \\
          0 & C_{42} - C_{41} & C_{43} - C_{41} &   
\end{pmatrix}
\]

Далее, при редуцировании первого столбца проверяются неравенства $\Mat[2,1] > \Mat[3,1]$ и $\Mat[3,1] > \Mat[4,1]$.
Нам известно, что 
$\Mat[2,1] = C_{21} - C_{23}$, $\Mat[3,1] = C_{31} - C_{34}$, $\Mat[4,1] = C_{41} - C_{41} = 0$.
Следовательно, проверяются неравенства
$C_{21} - C_{23} > C_{31} - C_{34}$ и $C_{31} - C_{34} > 0$.
Каждое из них удовлетворяет условиям леммы~\ref{lem:1}.

При редуцировании оставшихся трех столбцов ситуация повторяется.
Значение \Sum при редуцировании столбцов не меняется, так как каждый столбец уже содержит нули.

После этого в алгоритме~\ref{alg} выполняется проверка условия $\Sum \ge \Value$. Но $\Value = \infty$.
Поэтому алгоритм переходит к вычислению функции \ChooseArc.

Первым нулевым элементом является $\Mat[1,2]$. 
После этого в строке~\ref{alg:imin} алгоритма~\ref{alg:arc} выполняются сравнения $\infty > \Mat[1,3]$ и $\Mat[1,3] > \Mat[1,4]$. При этом, после предыдущего этапа редукции, имеем $\Mat[1,3] = C_{13} - C_{12}$ и $\Mat[1,4] = C_{14} - C_{12}$. Очевидно, неравенство $C_{13} - C_{12} > C_{14} - C_{12}$ удовлетворяет условиям леммы~\ref{lem:1}.
На этом шаге выполняется присвоение $m \ot C_{14} - C_{12}$.
Далее, в строке~\ref{alg:jmin} алгоритма~\ref{alg:arc} выполняются сравнения $\infty > \Mat[3,2]$ и $\Mat[3,2] > \Mat[4,2]$. При этом $\Mat[3,2] = C_{32} - C_{34}$ и $\Mat[4,2] = C_{42} - C_{41}$.
Условия леммы~\ref{lem:1} снова выполнены.
На этом шаге выполняется присвоение $k \ot C_{42} - C_{41}$.
Далее выполняется сравнение $m + k > -1$ или, что то же самое, $C_{14} - C_{12} + C_{42} - C_{41} > -1$.
Очевидно, это неравенство совместимо с условием $C \in K(\bm{y})$.
В переменную $w$ заносится значение выражения $C_{14} - C_{12} + C_{42} - C_{41}$.

Второй нулевой элемент "--- $\Mat[2,3]$.
Действуя по аналогии, перечислим только нетривиальные сравнения.
Неравенство $\Mat[2,1] \le \Mat[2,4]$ или $C_{21} - C_{23} \le C_{24} - C_{23}$, очевидно, совместимо с условием $C \in K(\bm{y})$.
Неравенство $\Mat[1,3] \le \Mat[4,3]$ тоже совместимо.
Далее, в строке~\ref{alg:mkw} проверяется неравенство $m+k > w$ или, с учетом предыдущих действий,
\[
C_{21} - C_{23} + C_{13} - C_{12} > C_{14} - C_{12} + C_{42} - C_{41}.
\]
Очевидно, оно удовлетворяет условиям леммы~\ref{lem:1}.
После этого шага 
\[
w = C_{21} - C_{23} + C_{13} - C_{12}.
\]

Третий нулевой элемент "--- $\Mat[3,4]$.
Неравенство $\Mat[3,1] < \Mat[3,2]$ или $C_{31} - C_{34} < C_{32} - C_{34}$, очевидно, совместимо с условием $C \in K(\bm{y})$.
Неравенство $\Mat[1,4] < \Mat[2,4]$ тоже совместимо.
Условие $m+k < w$ имеет вид
\[
C_{14} - C_{12} + C_{31} - C_{34} < C_{21} - C_{23} + C_{13} - C_{12}
\]
и тоже совместимо с условием $C \in K(\bm{y})$.

Четвертый нулевой элемент "--- $\Mat[4,1]$.
Легко проверить, что $\Mat[4,2] < \Mat[4,3]$ и $\Mat[3,1] < \Mat[2,1]$ совместимы с условием $C \in K(\bm{y})$.
Условие $m+k < w$ имеет вид
\[
C_{31} - C_{34} + C_{42} - C_{41} < C_{21} - C_{23} + C_{13} - C_{12}
\]
и тоже совместимо.

В данный момент мы все еще находимся в первом экземпляре процедуры \BranchBound.
После описанного выше выполнения функции \ChooseArc выбирается дуга $(i,j) = (2,3)$ (сумма $m+k$ для нее оказалась наибольшей),
из матрицы \Mat вычеркиваются 2-я строка и 3-й столбец, а дуга $(3,2)$ становится запрещенной.
На вход второго экземпляра процедуры \BranchBound подается матрица 
\[
C' \ot
\begin{pmatrix}
  & 0 &   & 1 \\
  &   &   &   \\
1 &   &   & 0 \\
0 & 1 &   &   
\end{pmatrix}
\]
(Пустая строка и пустой столбец оставлены для удобства чтения.)
Ясно, что при её редуцировании ничего нового не происходит, так как каждая строка и каждый столбец содержат нули. 
При вызове функции \ChooseArc в строке~\ref{alg:mkw} выполняются следующие сравнения типа $m+k > w$.
\[
C_{14} - C_{12} + C_{42} - C_{41} > -1.
\]
Очевидно, это неравенство совместимо с условием $C \in K(\bm{y})$.
Далее, выполняется неравенство
\[
C_{31} - C_{34} + C_{14} - C_{12} \le C_{14} - C_{12} + C_{42} - C_{41},
\]
которое удовлетворяет условиям леммы~\ref{lem:1}.
Следующее сравнение
\[
C_{31} - C_{34} + C_{42} - C_{41} \le C_{14} - C_{12} + C_{42} - C_{41}
\]
тоже совместимо с $C \in K(\bm{y})$.

Итак, после вызова функции \ChooseArc во втором экземпляре \BranchBound, выбирается дуга $(1,2)$.
Гамильтонов цикл с дугами $(2,3)$ и $(1,2)$ определяется однозначно.
Выполняется присвоение
\[
\Value \ot C_{12} + C_{23} + C_{34} + C_{41}.
\]
После этого алгоритм переходит к рассмотрению случаев, когда контур содержит дугу $(2,3)$, но не содержит $(1,2)$. Запускается третий экземпляр \BranchBound с матрицей
\[
C'' \ot
\begin{pmatrix}
  &   &   & 1 \\
  &   &   &   \\
1 &   &   & 0 \\
0 & 1 &   &   
\end{pmatrix}
\]
При редуцировании две единицы заменяются нулями. Никакие <<отбрасывающие>> сравнения не выполняются.
Значение переменной \Sum увеличивается на $\Mat[1,4] = C_{14} - C_{12}$ и на $\Mat[4,2] = C_{42} - C_{41}$.
Текущий экземпляр процедуры завершается в строке~\ref{alg:return} после проверки неравенства
$\Sum \ge \Value$:
\[
(C_{14} - C_{12}) + (C_{42} - C_{41}) > 0.
\]
Заметим, что допустимое решение $\bm{y}$ полностью отбраковывается алгоритмом именно на этом шаге (с учетом ранее проверенного неравенства $C_{31} > C_{34}$).
Тем не менее, это неравенство удовлетворяет условиям леммы~\ref{lem:1} и, следовательно, совместно с условием $C \in K(\bm{y})$.

Вместе с третьим экземпляром процедуры \BranchBound завершается и второй её экземпляр.
Алгоритм переходит к выполнению предпоследней строки в первом экземпляре.
В этом экземпляре
\[
\Sum =  C_{12} + C_{23} + C_{34} + C_{41}.
\]
Для разбора случаев, когда контур не содержит дугу $(2,3)$ вызывается четвертый экземпляр процедуры с матрицей
\[
C''' \ot
\begin{pmatrix}
  & 0 & 2 & 1 \\
2 &   &   & 2 \\
1 & 2 &   & 0 \\
0 & 1 & 2 &   
\end{pmatrix}
\]
При редуцировании второй строки выполняется сравнение $\Mat[2,1] \le \Mat[2,4]$.
При редуцировании третьего столбца "--- $\Mat[1,3] \le \Mat[4,3]$.
Очевидно, ни то ни другое не отбрасывают целиком конус $K(\bm{y})$.
Значение \Sum увеличивается на $(C_{21} - C_{23}) + (C_{13} - C_{12})$.

И, наконец, сравнение $\Sum \ge \Value$ завершает этот четвертый экземпляр процедуры и вообще весь алгоритм. Это сравнение имеет вид
\[
(C_{21} - C_{23}) + (C_{13} - C_{12}) \ge 0
\]
и тоже совместимо с условием $C \in K(\bm{y})$.

Итак, условие (*) из определения~\ref{def:direct-type2} не выполнено для этого алгоритма.

%
%

\medskip

Лаборатория <<Дискретная и вычислительная геометрия>>, ЯрГУ им. П.Г. Демидова, ул.~Советская 14, Ярославль, 150000. E-mail: \verb|maximenko.a.n@gmail.com|

\end{document}